\documentclass[11pt]{article}
\usepackage{fullpage}
\usepackage{amsmath,amsfonts,amsthm,amssymb,bbm}
\usepackage[colorlinks]{hyperref}
	\hypersetup{linkcolor=[rgb]{.7,0,.7}}
	\hypersetup{citecolor=[rgb]{.5,0,.5}}
	\hypersetup{urlcolor=[rgb]{.7,0,.7}}
\usepackage[capitalise]{cleveref}
\usepackage{tikz}
\usepackage{xcolor}
\usepackage{titlesec}

\newcommand{\mc}{\mathcal}
\newcommand{\E}{\mathop{\mathbf{E}\mbox{}}\limits}

\newcommand{\one}{\mathbf{1}}

\newcommand{\supp}{\mathrm{supp}}
\newcommand{\val}{\mathrm{val}}

\newcommand{\Mid}{\mathrel{\Big|}}

\newcommand{\set}[1]{\left\{ #1 \right\}}	
\newcommand{\ts}{\textsuperscript}

\newtheorem{theorem}{Theorem}[section]
\newtheorem{lemma}[theorem]{Lemma}
\newtheorem{proposition}[theorem]{Proposition}
\newtheorem{corollary}[theorem]{Corollary}
\newtheorem{claim}[theorem]{Claim}

\theoremstyle{remark}

\theoremstyle{definition}
\newtheorem{definition}[theorem]{Definition}

\begin{document}
\title{Polynomial Bounds On Parallel Repetition For All 3-Player Games With Binary Inputs}


\author{Uma Girish\thanks{\scriptsize Princeton University. E-mail: \href{ugirish@cs.princeton.edu}{\texttt{ugirish@cs.princeton.edu}.} Research supported by the Simons Collaboration on Algorithms and Geometry, by a Simons Investigator Award, by the National Science Foundation grants No. CCF-1714779, CCF-2007462 and by the IBM Phd Fellowship.} \and Kunal Mittal\thanks{\scriptsize Princeton University. E-mail: \href{kmittal@cs.princeton.edu}{\texttt{kmittal@cs.princeton.edu}.} Research supported by the Simons Collaboration on Algorithms and Geometry, by a Simons Investigator Award and by the National Science Foundation grants No. CCF-1714779, CCF-2007462.} \and Ran Raz\thanks{\scriptsize Princeton University.  E-mail: \href{ranr@cs.princeton.edu}{\texttt{ranr@cs.princeton.edu}.}  Research supported by the Simons Collaboration on Algorithms and Geometry, by a Simons Investigator Award and by the National Science Foundation grants No. CCF-1714779, CCF-2007462.} \and Wei Zhan\thanks{\scriptsize Princeton University.  E-mail: \href{weizhan@cs.princeton.edu}{\texttt{weizhan@cs.princeton.edu}.} Research supported by the Simons Collaboration on Algorithms and Geometry, by a Simons Investigator Award and by the National Science Foundation grants No. CCF-1714779, CCF-2007462.}}
\date{}
\maketitle

\begin{abstract}
	We prove that for every 3-player (3-prover) game $\mathcal G$ with value less than one, whose query distribution has the support $\mathcal S = \{(1,0,0), (0,1,0), (0,0,1)\}$ of hamming weight one vectors, the value of the $n$-fold parallel repetition $\mathcal G^{\otimes n}$ decays polynomially fast to zero; that is, there is a constant $c = c(\mathcal  G)>0$ such that the value of the game $\mathcal  G^{\otimes n}$ is at most $n^{-c}$.
	
	Following the recent work of Girish, Holmgren, Mittal, Raz and Zhan (STOC 2022), our result is the missing piece that implies a similar bound for a much more general class of multiplayer games:
	For \textbf{every} 3-player game $\mathcal G$ over \emph{binary questions} and \emph{arbitrary answer lengths}, with value less than 1, there is a constant $c = c(\mathcal G)>0$ such that the value of the game $\mathcal G^{\otimes n}$ is at most $n^{-c}$.
	
	Our proof technique is new and requires many new ideas.
	For example, we make use of the Level-$k$ inequalities from Boolean Fourier Analysis, which, to the best of our knowledge, have not been explored in this context prior to our work.
\end{abstract}

\newpage

\section{Introduction}
Our main object of study is multiplayer (multiprover) games.
A $k$-player game $\mc G$ consists of $k$ players who are playing against a referee.
The game begins by the referee sampling a $k$-tuple of questions $(x^1, \dots, x^k)$ from some global distribution $Q$.
The referee then gives the question $x^j$ to the $j\ts{th}$ player, for each $j\in [k]$, based on which they give back an answer $a^j$.
Finally, the referee evaluates a predicate $V((x^1, \dots, x^k), (a^1, \dots, a^k))$ and says that the players win if and only if the predicate evaluates to true.
The value $\val(\mc G)$ of the game is defined to be the maximum  winning probability for the players.
Note that the probability here is over the randomness used by the referee to sample $(x^1, \dots, x^k)\sim Q$, and the maximum is over the strategies used by the players.

Given a game $\mc G$ with value $\val(\mc G)<1$, it is natural to consider the parallel repetition of the game $\mc G$, defined as follows: In the $n$-fold repetition $\mc G^{\otimes n}$ of the game $\mc G$, the referee independently samples questions for $n$ copies of the game $\mc G$; that is, the referee samples $(x^1_i, \dots, x^k_i)\sim Q$ independently for $i\in [n]$.
Then, the referee \emph{simultaneously} gives questions $x_1^j,\dots, x_n^j$ to the $j\ts{th}$ player, for each $j\in[k]$, who then gives back answers $a_1^j,\dots, a_n^j$.
The players are said to win the game if for each $i\in [n]$, the predicate $V((x_i^1, \dots, x_i^k), (a_i^1, \dots, a_i^k))$ evaluates to true.

With the above definition of the $n$-fold repeated game $\mc G^{\otimes n}$, it is interesting to study the behavior of $\val(\mc G^{\otimes n})$ with respect to $n$, and the initial parameters of the game $\mc G$~\cite{FRS94}.
Observe that $\val(\mc G^{\otimes n})\geq \val(\mc G)^n$, since any strategy that achieves value $\val(\mc G)$ in the game $\mc G$, when repeated independently for all copies $i\in [n]$, achieves the value $\val(\mc G)^n$ in the game $\mc G^{\otimes n}$.
While one would expect such an inequality to be tight, this is far from true; there are games such that $\val(\mc G^{\otimes n})$ is exponentially larger (with respect to $n$) compared to $\val(\mc G)^n$.
The crucial reason why this can happen is that in the game $\mc G^{\otimes n}$ the players are allowed to correlate their answers among different copies $i\in [n]$ of the game.
That is, it is not necessary (and not optimal) for every player's answer for the $i\ts{th}$ copy of the game to depend only on the $i\ts{th}$ question they receive.

Nevertheless, Raz \cite{Raz98} proved that for any \emph{2-player} game $\mc G$ with $\val(\mc G)<1$, it holds that $\val(\mc G^{\otimes n}) = 2^{-\Omega(n)}$.
This, and related techniques and results, turned out to be sufficient for a large number of applications: in the theory of interactive proofs \cite{BOGKW88}, PCPs and hardness of approximation \cite{BGS98, Fei98, Has01}, geometry of foams \cite{FKO07, KORW08, AK09}, quantum information \cite{CHTW04}, and communication complexity \cite{PRW97, BBCR13, BRWY13}.
The reader is referred to this survey \cite{Raz10} for more details.
There have been many subsequent improvements that improve the constants in the bounds, and even get better bounds based on the value $\val(\mc G)$ of the initial game \cite{Hol09, Rao11, BRRRS09, RR12, DS14, BG15}.

The case of 2-player games, hence,  is fairly well-understood with regards to the operation of parallel repetition.
On the other hand, despite much effort, the general question of parallel repetition for multiplayer games remains wide open.
The only general bound, by \cite{Ver96}, that applies to all $k$-player games, says that if $\val(\mc G)<1$, then $\val(\mc G^{\otimes n}) \leq \frac{1}{\alpha(n)}$, where $\alpha(n)$ is a function which grows like the (extremely slowly growing) inverse Ackermann function.
The weak bounds here result from a black-box use of the Density Hales-Jewett Theorem \cite{FK91, Pol12} from extremal combinatorics.

While there are some known potential applications of bounds on parallel repetition of multiplayer games, for example, \cite{MR21} show a connection between parallel repetition and super-linear lower bounds for non-uniform Turing machines, we believe that the notion of parallel repetition is so basic that it deserves attention in its own right.
As mentioned by \cite{DHVY17}, there are many problems in complexity theory that are inherently high dimensional, and which share this sudden difficulty of being tractable beyond dimension two.
For example, whereas direct sum and direct product theorems are known for two-party communication complexity, no such results are known for multiparty communication complexity in the number-on-forehead model (which is deeply related to proving new lower bounds in circuit complexity), for seemingly similar reasons to why there has not been much progress on multiplayer parallel repetition. 

Recent work, however, has made some progress on proving parallel repetition bounds for some special classes of multiplayer games:
\begin{enumerate}
	\item \textbf{Connected Games:}\label{intro:conn_game} Dinur, Harsha, Venkat and Yuen \cite{DHVY17} consider games which satisfy a certain connectedness property and show that the value of any such game satisfies an exponential decay bound under parallel repetition:  if $\val(\mc G)<1$ then $\val(\mc G^{\otimes n}) = 2^{-\Omega(n)}$. A $k$-player game $\mc G$ is said to have this connectedness property if the graph $\mc H_\mc G$ defined as follows is connected: The vertices of the graph are all the possible $k$-tuples of questions to the players (which are asked with non-zero probability), and there is an edge between two such $k$-tuples if they differ in the question to exactly one of the $k$-players.
	
	The proof for these games uses information theoretic techniques, and builds on the works on 2-player games by \cite{Raz98, Hol09}.
	
	\item \textbf{The GHZ Games:}\label{intro:ghz_game} \cite{HR20, GHMRZ21} show that any game $\mc G$ over the set of questions $\set{(x,y,z)\in \set{0,1}^3 : x+ y + z = 0\pmod{2} }$ satisfies a polynomial bound on the value of parallel repetition: if $\val(\mc G)<1$ then $\val(\mc G^{\otimes n}) = n^{-\Omega(1)}$. For such games, all vertices in the graph $\mc H_\mc G$ (as defined above in point~\ref{intro:conn_game}) are isolated, and the techniques of \cite{DHVY17} fail to be applicable.
	
	The known proofs for this case use Fourier analytic techniques that crucially rely on the fact that the inputs to the players define a linear subspace of $\mathbb{F}_2^3$.
	
	\item A recent work \cite{GHMRZ22} considers the problem of parallel repetition for 3-player games in which each player is asked a binary question. They do a case analysis of all such games and divide the general problem into the following cases:
	\begin{enumerate}
		\item Connected games or games that are essentially 2-player games: An exponential decay bound is known \cite{Raz98, DHVY17}.
		\item Games over the question set of the GHZ game (see point~\ref{intro:ghz_game}): A polynomial decay bound is known.
		\item Games over the question set $\set{(x,y,z)\in \set{0,1}^3 : x+y+z\not=2}$: They show that such games fall into a class of games which they call \textbf{playerwise connected games}, a generalization of the class of connected games. Informally, a game $\mc G$ is said to be playerwise connected if the \emph{projection} of the graph $\mc H_\mc G$ onto each of the $k$-players is connected. They show that any playerwise connected game satisfies a polynomial decay bound in the value of parallel repetition: if $\val(\mc G)<1$ then $\val(\mc G^{\otimes n}) = n^{-\Omega(1)}$.
		\item \label{intro:4_pt_and} Games over the question set $\set{(x,y,z)\in \set{0,1}^3 : z = x y}$: They call this the four-point AND distribution, and show that any such game satisfies a polynomial bound in the value of parallel repetition.
		\item \label{intro:3_pt_case} Games over the set of questions $\mc S = \set{(1,0,0), (0,1,0), (0,0,1)}$ of hamming-weight one: They do not prove a general bound for games in this class, but rather only for games where the answers given by each of the three players is also binary. Under this extra assumption, they are in fact able to prove an exponential decay bound under parallel repetition.
		
		A very interesting game which they consider is the \textbf{anti-correlation game} defined as follows: The referee samples the questions $(x^1,x^2,x^3) \in \mc S$ uniformly at random, and the two players who receive the input 0 must produce different outputs in $\set{0,1}$. This game has the special property that while its \emph{non-signalling} value is less than 1, the non-signalling value does not decrease at all under parallel repetition \cite{HY19}.
	\end{enumerate}
\end{enumerate}

The main topic of interest of the current paper are games described above in point~\ref{intro:3_pt_case}, that is, all games over the question set $\mc S = \set{(1,0,0), (0,1,0), (0,0,1)}$.
The work \cite{GHMRZ22} shows that any bounds for a special subclass of such games qualitatively translate to the same bounds for all games in this class.
In particular, polynomial decay bounds for the value of parallel repetition for the following subclass of games implies polynomial decay bounds for all games over the question set $\mc S = \set{(1,0,0), (0,1,0), (0,0,1)}$:

\begin{definition}\label{intro:main_game_def}
	Let $k\in\mathbb{N}$, and let $\mc{S}=\{(1,0,0),(0,1,0),(0,0,1)\}$.
	We define a 3-player game $\mc{G}_k$ on 3 players Alice, Bob and Charlie as follows:
	\begin{enumerate}
		\item The referee samples $(x,y,z)\in \mc S$ uniformly at random, and gives $x,y,z$ to the three players respectively.
		\item The players answer $a\in \set{0,1}^k, b\in \set{0,1}^k, c\in [k]$ respectively.
		\item The winning predicate is defined as:
		\[ V_k((x,y,z),(a,b,c)) = \begin{cases}
			b_c = 0, & \textrm{ if }(x,y,z)=(1,0,0) \\
		a_c = 0, & \textrm{ if }(x,y,z)=(0,1,0) \\
		\forall i\in[k], a_i+ b_i\geq 1,  & \textrm{ if }(x,y,z)=(0,0,1) \end{cases}. \]
	\end{enumerate}
    In other words, two randomly chosen players receive 0 as input and the third player gets a 1 as input.
    The predicate only depends on the two players who get 0 as input, and only those two players play the game.
    If Charlie and Alice (or Bob) are playing, Charlie must point to an index where Alice (or Bob) outputs 0.
    On the other hand, if Alice and Bob are playing, they must each output $k$-bit strings such that the bit-wise-OR of the two strings is the all 1s string.
\end{definition}

Our main result is a polynomial decay bound on the parallel repetition for all games in the above subclass:

\begin{theorem}\label{thm:intro_main}
There exists an absolute constant $c>0$ such that the following holds:
For every $k\in\mathbb{N}$, and for every sufficiently large $n\in \mathbb{N}$, it holds that $\val(\mc G_k^{\otimes n}) \leq n^{-c}$, where the game $\mc G_k$ is as defined in Definition~\ref{intro:main_game_def}.
\end{theorem}

Based on the previous discussion, combined with the works \cite{Raz98, DHVY17, HR20, GHMRZ21, GHMRZ22}, our theorem implies the following:

\begin{theorem}
	Let $\mc G$ be any 3-player game over binary questions, and arbitrary finite length answers, such that $\val(\mc G)<1$.
	Then, there exists a constant $c = c(\mc G)>0$, such that for every $n\in \mathbb{N}$, it holds that $\val(\mc G_k^{\otimes n}) \leq n^{-c}$.
\end{theorem}

We remark that Hazla, Holenstein and Rao \cite{HHR16} consider games over the same question set $\mc S = \set{(1,0,0), (0,1,0), (0,0,1)}$, and show barriers for proving parallel repetition bounds for such games using the \emph{forbidden subgraph method} \cite{FV02}.
Our result builds new techniques that do not fit into the above framework, and are able to bypass these barriers.
 
Next, in Section~\ref{sec:pf_overview}, we give an overview of the proof of Theorem~\ref{thm:intro_main}.
We note that our proof introduces several new ideas, which we believe are very general and can possibly extend to much larger classes of games.
For example, in one of the steps, we use Fourier Analysis over the boolean hypercube, and in particular the Level-$k$ inequalities; to the best of our knowledge, such use in the context of parallel repetition is new.

	
	\subsection{Proof Overview}\label{sec:pf_overview}
	Fix some $k\in \mathbb{N}$ and consider the $n$-fold repeated game $\mc G_k^{\otimes n}$ (see Definition~\ref{intro:main_game_def}). We'll use the term \emph{coordinate} to mean a tuple $(i,j)$ with $i\in [n]$ and $j\in [k]$, that indexes an answer for Alice or Bob. Recall that in each copy of the game $\mc{G}_k$, only the two players who receive input 0 affect the winning predicate, and we say that they are the ones who \emph{play}.
	
	The high-level intuition  is as follows: In order to win, Alice and Bob cannot both answer 0 at the same coordinate. On the other hand, suppose that they indeed only answer 0 in two fixed disjoint subsets of coordinates each of their own, then Charlie's answer in each copy of the game $\mc{G}_k$ actually reveals which player he is playing with, which is too much information for Charlie to have.
	
	We note, however, that this intuition is too simplistic and the actual proof is much more complicated, because in each coordinate only two out of the 3 players play. Nevertheless, our proof can be viewed as a rigorous execution of the intuition, by finding a large enough product event $E_1\times E_2$ on Alice's and Bob's inputs in which the above presumption holds true. More specifically, to prove by contradiction we assume that the winning probability is at least $n^{-c}$ (where $c>0$ is a small constant), and the proof is carried out in three steps:
	
	\paragraph{Remove coordinates that Alice and Bob lose (\cref{section:step1}).} We remove the coordinates where Alice and Bob both play and simultaneously output 0 with non-negligible (at least $n^{-O(c)}$) probability, by fixing their inputs and outputs in these coordinates. The fixing of outputs gives rise to the product event $E_1\times E_2$ on the remaining coordinates. We need to ensure that the probability of both $E_2$ and the winning event $W$ remain $n^{-O(c)}$, while the rounds of removal are few so that $E_1$ is also not extremely small. This is done by a potential function argument that tracks both $P(E_2|E_1)$ and $P(W|E_1,E_2)$, while the latter has higher weight than the former in the potential function. The potential function is non-decreasing, and increases by a non-negligible amount every time we exclude the losing part by fixing, thus guaranteeing the above-mentioned requirements as the probabilities cannot exceed~1.
	
	We remark that proving a similar bound with only $P(W|E_1, E_2)$ being $n^{-O(c)}$, and $P(E_1, E_2)$ being $2^{-n^{O(c)}}$ is not too hard. However for the latter part of our proof, we need that $P(E_2|E_1)$ is also at least $n^{-O(c)}$. Hence, when removing coordinates, we fix the inputs and outputs in a very delicate manner, and analyze the evolution of potential function accordingly.
	
	\paragraph{Establish independence of Alice's and Bob's answers (\cref{section:step2}).} Now that in each coordinate, Alice and Bob rarely both simultaneously output 0, we would like to strengthen the claim so that in each coordinate either Alice or Bob answers 0 with negligible probability. In other words, in each coordinate their answers are close to being independent. For a fixed coordinate, we consider Alice's output as a boolean function of her input, and the average of her output given Bob's input is exactly the sum of Fourier coefficients in the subcube where Bob receives 1. If we take average over any large event for Bob, then every Fourier coefficient, except the first one, will contribute negligibly to the result, meaning Bob answering 0 is close to being independent of Alice's answer.
	
	This is not true, of course, unless Bob receives 1 with small enough probability. Fortunately the first step does not depend on the query distribution, and therefore we can change the query distribution at the very beginning, from uniform to the one where $(0,1,0)$ has probability close to (but still polynomially larger than) $1/n$. It turns out that the change of distribution does not affect the parallel repetition property. With the right distribution, we bound the contributions of the Fourier coefficients as claimed above using Level-$k$ inequalities.
	
	\paragraph{Bound winning probability for Charlie (\cref{section:step3}).} The previous steps indicate that Alice and Bob each owns a fixed set of coordinates where they output 0 with non-negligible probabilities, and the two sets are disjoint. Now consider an input $(x,y,z)$. Among the copies of games where Charlie needs to answer (Charlie receives 0), let $G_1$ (and $G_2$) be the copies where Charlie's answer points at a coordinate that Alice (and Bob) owns. On the other hand, in each coordinate they do not own, Alice and Bob output 0 with only negligible probability, so let $B_1$ (and $B_2$) be the copies where Alice's (and Bob's) answer string contains 0 outside the coordinates they own. Note that $B_1$ depends only on $x$, $B_2$ depends only on $y$, while $G_1$ and $G_2$ depend only on $z$.
	
	In order to win, $G_1\cup B_1$ have to cover all the copies that Alice plays with Charlie, which is the 1's in $y$, and $G_2\cup B_2$ have to cover all the copies that Bob plays with Charlie, which is the 1's in $x$. But for a typical input $(x,y,z)$, where both $|x|$ and $|z|$ are close to $n/2$, $G_1$ and $B_1$ intersect with the 1's in $y$ in proportion to their sizes. That means $G_1$ has to cover almost all the copies that Charlie plays, and thus $G_2\cup B_2$ is not large enough to cover the 1's in $x$, as $G_1$ and $G_2$ are disjoint while $B_1$ and $B_2$ are negligibly small. This contradicts the fact that the winning probability is high, even conditioned on $E_1\times E_2$.
	
	\section{Preliminaries}
	
	We use $\log$ to denote the logarithm under base $2$, with the convention that $\log 0=-\infty$. Let $\mathbb{N}=\{1,2,\ldots\}$ be the set of natural numbers. For every $n\in\mathbb{N}$, let $[n]$ be the set $\{1,2,\ldots,n\}$.
	
	For every $x\in\{0,1\}^n$, $i\in[n]$ and $S\subseteq[n]$, we use $x_i\in\{0,1\}$ to denote the bit on index $i$, and $x_S\in\{0,1\}^{|S|}$ to denote the substring of $x$ on $S$. Let $\one(x)\subseteq[n]$ be the set of indices $i$ where $x_i=1$, and let $|x|=|\one(x)|$ be the Hamming weight of $x$. We also define a partial order on $\{0,1\}^n$ such that $x\geq y$ if and only if $x_i=1$ whenever $y_i=1$.
	
	For a random variable $X$, we use $\supp(X)$ to denote its support. We define a \emph{fixing} of the random variable $X$ to be an event that assigns $X$ to be some fixed value in $\supp(X)$. We equate every subset $E\subseteq\supp(X)$ to an event on $X$. We use $P(E)$ to denote the probability of an event $E$ under the distribution $P$. 
	
	\begin{lemma}[Chernoff Bounds, see {\cite{MU05}}]\label{lemma:chernoff}
		Let $X_1,\ldots, X_n \in \{0,1\}$ be independent random variables each with mean $\mu$, and let $X = \sum_{i=1}^n X_i$. Then, for all $\delta \in (0,1)$, it holds that 
		\[\Pr[X \leq (1-\delta)\mu n] \leq e^{-\frac{\delta^2 \mu n}{2}},\]
		\[\Pr[X \geq (1+\delta)\mu n] \leq e^{-\frac{\delta^2 \mu n}{3}}.\]
	\end{lemma}
	
	\begin{lemma}\label{lemma:potential}
	    Let $P$ be a distribution and $A,B$ be two events such that $P(A\wedge B)>0$.
	    Let $X$ be a random variable with finite support, and let $\mathcal X = \{x: P(X=x | B)>0\}$, and let $x_0\in \mathcal X$ be a fixed element such that $P(X=x_0)\geq \delta$.
	    
	    For each $x\in \mathcal X$, we define $\Phi(x) = \log P(A|B,X=x)+\frac{1}{2}\log P(B|X=x)$, and let $\Phi = \log P(A|B)+\frac{1}{2}\log P(B) < 0$.
	    
	    Then, for every $0<\varepsilon<1$, it holds that either
		\[
			\Phi(x_0) \geq \Phi -\varepsilon,
		\]
		or
		\[
		    P\left(X\in \mathcal X\  \land\  \Phi(X)\geq \Phi + \frac{1}{8}\delta \varepsilon \right) \geq 2^{2 \Phi}\cdot\frac{1}{4}\delta^2\varepsilon  
		\]
	\end{lemma}

	\begin{proof}
		By Jensen's inequality, we have
		\begin{align*}
			P(A|B)^2\cdot P(B) & = 
			   \left(\sum_{x\in\mc{X}}P(X=x|B)\cdot P(A|B, X=x)\right)^2\cdot P(B) \\
			 & \leq \sum_{x\in\mc{X}}P(X=x|B)\cdot P(A|B, X=x)^2\cdot P(B) \\
			 & = \sum_{x\in\mc{X}}P(X=x)\cdot P(A|B,X=x)^2\cdot P(B|X=x).
		\end{align*}
		Suppose that $\Phi(x_0)<\Phi-\varepsilon$, which implies that
		\begin{align*}
			P(A|B,X=x_0)^2\cdot P(B|X=x_0) &<P(A|B)^2\cdot P(B)\cdot 2^{-2\varepsilon} \\
			&\leq P(A|B)^2\cdot P(B)\cdot(1-\varepsilon/4).
		\end{align*}
		On the other hand, since $\delta,\varepsilon\leq 1$ we have $\log(1+\delta\varepsilon/4)\geq \delta\varepsilon/4$, and thus in order to satisfy $\Phi(x)\geq\Phi+\frac{1}{8}\delta\varepsilon$ it suffices to have 
		\begin{equation}\label{eq}
		    P(A|B,X=x)^2\cdot P(B|X=x) \geq P(A|B)^2\cdot P(B)\cdot (1+\delta\varepsilon/4).
		\end{equation}
		Let $\mc{X}_1\subset\mc{X}$ be the set of $x\in\mc{X}$, $x\neq x_0$ that satisfies \eqref{eq}. Since $P(A|B,X=x)^2\cdot P(B|X=x)\leq 1$, we have
		\begin{align*}
		    &P\left(X\in \mathcal X\  \land\  \Phi(X)\geq \Phi + \frac{1}{8}\delta \varepsilon \right) \\
		    \geq\ & \sum_{x\in\mc{X}_1} P(X=x)\cdot P(A|B,X=x)^2\cdot P(B|X=x) \\
		    =\ & \sum_{x\in\mc{X}}P(X=x)\cdot P(A|B,X=x)^2\cdot P(B|X=x) 
		     - P(X=x_0)\cdot P(A|B,X=x_0)^2\cdot P(B|X=x_0) \\
		    & - \sum_{x\notin \mc{X}_1, x\neq x_0}P(X=x)\cdot P(A|B,X=x)^2\cdot P(B|X=x) \\
			\geq\ & P(A|B)^2\cdot P(B)\big[1-P(X=x_0)\cdot(1-\varepsilon/4)-P(X\neq x_0)\cdot (1+\delta\varepsilon/4)\big] \\
			\geq\ & P(A|B)^2\cdot P(B)\cdot \frac{1}{4}\delta^2\varepsilon. \qedhere
		\end{align*}
	\end{proof}
	
	\subsection{Fourier Analysis}
	For every $x,y\in\{0,1\}^n$, let $x\cdot y$ be their inner product in $\mathbb{Z}$. Given a function $f:\{0,1\}^n\rightarrow\mathbb{R}$, let $\widehat{f}:\{0,1\}^n\rightarrow\mathbb{R}$ be its Fourier coefficients, defined as
	\[\widehat{f}(u)=\frac{1}{2^n}\sum_{x\in\{0,1\}^n}(-1)^{x\cdot u}f(x).\]
	We will use the following equation on the sum of the Fourier coefficients in a subcube, which follows from Plancherel's theorem: For every $y\in\{0,1\}^n$, we have
	\[\sum_{u\leq y}\widehat{f}(u)=\frac{1}{2^{n-|y|}}\sum_{x\cdot y=0}f(x).\]
	We will also use the following version of the Level-$k$ inequality:
	\begin{lemma}\label{lemma:level}
		Let $f:\{0,1\}^n\rightarrow\{-1,0,1\}$ be a function with $\frac{1}{2^n}\sum_x|f(x)|=\alpha$. Then for every $k\in\mathbb{N}$,
		\[\sum_{|u|=k}|\widehat{f}(u)|\leq (2en\cdot \max\{1,\ln(1/\alpha)\})^{k/2}\cdot\alpha.\]
	\end{lemma}
	\begin{proof}
		Since there are at most $n^k$ many $u$ with $|u|=k$, we have
		\[\sum_{|u|=k}|\widehat{f}(u)|\leq n^{k/2}\sqrt{\sum_{|u|=k}\widehat{f}(u)^2}.\]
		Therefore it suffices to prove that
		\[\sum_{|u|=k}\widehat{f}(u)^2\leq (2e\cdot \max\{1,\ln(1/\alpha)\})^k\cdot\alpha^2.\]
		When $k\leq 2\ln(1/\alpha)$, it follows from the original Level-$k$ inequality (see \cite[Section~9.5]{Odo14}). When $k>2\ln(1/\alpha)$, we also have
		\[\sum_{u\in\{0,1\}^n}\widehat{f}(u)^2=\frac{1}{2^n}\sum_{x\in\{0,1\}^n} f(x)^2=\alpha\leq e^k\alpha^2\leq (2e\cdot \max\{1,\ln(1/\alpha)\})^k\cdot\alpha^2. \qedhere\]
	\end{proof}
	
	\subsection{Multi-player Games}
	
	The notations we use here follows mostly from \cite{GHMRZ22}.
	
	\begin{definition} (Multiplayer Game)
		A $k$-player game $\mc{G}$ is a tuple $\mc{G}=(\mc{X},\mc{A},Q,V)$, where the question set $\mc{X} = \mc{X}^1 \times\cdots\times \mc{X}^k$, and the answer set $\mc{A}= \mc{A}^1 \times\cdots\times \mc{A}^k$ are finite sets, $Q$ is a probability distribution over $\mc{X}$, and $V:\mc{X}\times \mc{A}\to\{0,1\}$ is a predicate.
	\end{definition}
	
	\begin{definition} (Game Value)
		Let $\mc{G}=(\mc{X},\mc{A},Q,V)$ be a $k$-player game. The value $\val(\mc{G})$ of the game $\mc{G}$ is defined as 
		\[\val(\mc{G}) = \max_{f^1,\ldots,f^k}\ \Pr_{X\sim Q}\Big(V(X,(f^1(X^1),\ldots,f^k(X^k)))=1\Big),\]
		where the maximum is over all sequence of functions $\left(f^j:\mc{X}^j\rightarrow\mc{A}^j\right)_{j\in[k]}$, which we call player strategies.
	\end{definition}
	
	We note that the value of the game is unchanged even if we allow the player strategies to be randomized, that is, we allow the strategies to depend on some additional shared and private randomness. 
	
	\begin{definition} (Parallel Repetition of a game)
		Let $\mc G = (\mc X, \mc A, Q, V)$ be a $k$-player game. We define its $n$-fold repetition as $\mc G^{\otimes n} = (\mc X^{\otimes n}, \mc A^{\otimes n}, P, V^{\otimes n})$.
		The sets $\mc X^{\otimes n}$ and $\mc A^{\otimes n}$ are defined to be the $n$-fold product of the sets $\mc X$ and $\mc A$ with themselves respectively.
		The distribution $P$ is the $n$-fold product of the distribution $Q$ with itself, that is, $P(X=x) = \prod_{i=1}^n Q(X_i=x_i)$.
		The predicate $V^{\otimes n}$ is defined as $V^{\otimes n}(x, a) = \bigwedge_{i=1}^n V(x_i, a_i)$.
	\end{definition}
	
	In this paper we mostly deal with 3-player games, and we use the notation $\mc{G}=(\mc{X}\times\mc{Y}\times\mc{Z},\mc{A}\times\mc{B}\times\mc{C},Q,V)$. That is, we use $\mc{X},\mc{Y},\mc{Z}$ in places of $\mc{X}^1,\mc{X}^2,\mc{X}^3$ and use $\mc{A},\mc{B},\mc{C}$ in places of $\mc{A}^1,\mc{A}^2,\mc{A}^3$. We also refer to the three players as Alice, Bob and Charlie.
	
	The proof of the following useful lemma is essentially the same as Lemma 3.14 in \cite{GHMRZ22}.
	\begin{lemma}\label{lemma:change}
		Let $\mc{G}_1=(\mc{X},\mc{A},Q_1,V)$ and $\mc{G}_2=(\mc{X},\mc{A},Q_2,V)$ be two multi-player games where only the distributions are different. Let $\lambda\in[0,1]$ be such that for every $x\in\mc{X}$, $Q_1(X=x)\geq\lambda Q_2(X=x)$. Then for every $n\in\mathbb{N}$, it holds that
		\[\val(\mc{G}_1^{\otimes n})\leq e^{-\lambda n/8}+\val(\mc{G}_2^{\otimes \lfloor \lambda n/2\rfloor}).\]
	\end{lemma}
    \begin{proof}
	Notice that we can write $Q_1 = \lambda Q_2 + (1-\lambda)Q'$ for some distribution $Q'$ over $\mc X$. 
	Let $Z = (Z_1,\dots, Z_n) \in \set{0,1}^n$ be i.i.d. Bernoulli random variables such that for each $i\in [n]$, independently, $Z_i$ is 1 with probability $\lambda$ and 0 with probability $1-\lambda$.
	For each $i\in[n]$, we think of the $i$-th copy of $Q_1$ as depending on $Z_i$: if $Z_i=1$ then $Q_1$ is drawn from $Q_2$, otherwise $Q_1$ is drawn from $Q'$.
	
	In order to bound the value of the game $\mc{G}_1^{\otimes n}$, we can assume that each of the players is also given $Z$ as input, since this can only increase the game's value. Observe that conditioned on the event $Z=z$ for any fixed value $z\in \set{0,1}^n$, the value of the game is at most the value of $\mc{G}_2^{\otimes |z|}$. Thus we have
	\begin{align*}
		\val(\mc G^{\otimes n}) & \leq \sum_{m=0}^n\Pr[|Z| = m]\cdot \val(\mc{G}_2^{\otimes m})
		\\& \leq \Pr\left[|Z|\leq \frac{\lambda n}{2}\right]\cdot 1 + 1\cdot \val(\mc{G}_2^{\otimes \lfloor \lambda n/2\rfloor})
		\\& \leq e^{-\lambda n/8} + \val(\mc{G}_2^{\otimes \lfloor \lambda n/2\rfloor}). \qedhere
	\end{align*}
    \end{proof}
	
	\section{Main Results}
	
	\begin{definition}\label{def:game}
		Let $U$ be the uniform distribution over $\mc{S}=\{(1,0,0),(0,1,0),(0,0,1)\}$. For every $k\in\mathbb{N}$ and every distribution $Q$ over $\mc{S}$, we define a 3-player game $\mc{G}_k(Q)=(\mc{X}\times\mc{Y}\times\mc{Z},\mc{A}_k\times\mc{B}_k\times\mc{C}_k,Q,V_k)$ with $\mc{X}=\mc{Y}=\mc{Z}=\{0,1\}$ as follows:
		\begin{itemize}
			\item[(a)] $\mc{A}_k=\mc{B}_k=\{0,1\}^k$ and $\mc{C}_k=[k]$.
			\item[(b)] For all $(x,y,z)\in\mc{S}$ and $(a,b,c)\in\mc{A}_k\times\mc{B}_k\times\mc{C}_k$,
			\[V_k((x,y,z),(a,b,c))=\left\{\begin{array}{ll}
			b_c = 0, & \textrm{ if }(x,y,z)=(1,0,0) \\
			a_c = 0, & \textrm{ if }(x,y,z)=(0,1,0) \\
			\forall i\in[k], a_i+b_i\geq 1,  & \textrm{ if }(x,y,z)=(0,0,1)
			\end{array}\right.\]
		\end{itemize}
	\end{definition}
	
	\begin{theorem}\label{thm:main}
		For every $k\in\mathbb{N}$, there exists $N_k\in\mathbb{N}$ such that for every $n\in\mathbb{N}, n\geq N_k$, it holds that $\val(\mc{G}_k(U)^{\otimes n})\leq n^{-1/2000}$.
	\end{theorem}
	
	Based on the results in \cite[Section 8.2]{GHMRZ22} and our discussions in the introduction, \cref{thm:main} implies the following bound on the parallel repetitions of 3-player games with binary inputs:
	\begin{theorem}
		Let $\mc{G}=(\mc{X}\times\mc{Y}\times\mc{Z},\mc{A}\times\mc{B}\times\mc{C},Q,V)$ be any 3-player game with $\mc{X}=\mc{Y}=\mc{Z}=\{0,1\}$, and such that $\val(\mc{G})<1$. Then there exists a constant $c=c(\mc{G})>0$ such that for every $n\in\mathbb{N}$, it holds that $\val(\mc{G}^{\otimes n})\leq n^{-c}$.
	\end{theorem}
	
	The rest of our paper is devoted to proving \cref{thm:main}. 
	
	\subsection{Change the distribution}
	
	In order to prove \cref{thm:main}, from now on we assume $\val(\mc{G}_k(U)^{\otimes n_1})\geq n_1^{-1/2000}$ for some large enough $n_1\in\mathbb{N}$, and eventually derive a contradiction. The first thing to do is changing the distribution so that Bob gets input $1$ with small probability.
	
	\begin{definition}\label{def:distr}
		Let $n=\lfloor n_1/3\rfloor$ and $c=1/1000$. Let $Q$ be the distribution over $\mc{S}$ such that $(0,1,0)$ has probability $n^{-1+100c}$, while $(1,0,0)$ and $(0,0,1)$ both have probability $\frac{1}{2}-\frac{1}{2}n^{-1+100c}$ each.
	\end{definition}
	
	\begin{claim}\label{claim:main}
		 $\val(\mc{G}_k(Q)^{\otimes n})\geq n^{-c}$.
	\end{claim}
	\begin{proof}
		Let $\lambda=2/3$, and thus we have $1/3\geq\lambda Q((X,Y,Z)=(x,y,z))$ for all $(x,y,z)\in\mc{S}$. Applying \cref{lemma:change} on $\mc{G}_k(U)$ and $\mc{G}_k(Q)$ gives
		\begin{align*}
			\val(\mc{G}_k(Q)^{\otimes \lfloor n_1/3\rfloor})
			 &\geq \val(\mc{G}_k(U)^{\otimes n_1})-e^{-\lambda n_1/8} \\
			 &\geq n_1^{-c/2}-e^{-n/4} \geq n^{-c}. \qedhere
		\end{align*}
	\end{proof}

	Let $P$ be the distribution $Q^{\otimes n}$, and let $(X,Y,Z)\in\mc{S}^n$ be the random variables that represent the inputs to the three players under distribution $P$. Let $f,g:\{0,1\}^n\rightarrow\{0,1\}^{n\times k}$ and $h:\{0,1\}^n\rightarrow[k]^n$ be strategies that achieve the value $\val(\mc{G}_k(Q)^{\otimes n})$, and let $W$ be the event that $(f,g,h)$ wins on the inputs $(X,Y,Z)$, so that we have $P(W)\geq n^{-c}$.
	
	\section{Remove Coordinates with $(0,0)$ Answers}\label{section:step1}
	
	\begin{lemma}\label{lemma:event}
		There exist $S\subseteq[n]$, a fixing $F$ of $(X_S,Y_S,Z_S)$, and two events $E_1\subseteq \mc{X}^{\otimes n}, E_2\subseteq \mc{Y}^{\otimes n}$ for Alice and Bob respectively, such that the following holds:
		\begin{itemize}
			\item[(a)] $|S|\leq n^{28c}$ and $P(E_1|F)\geq e^{-n^{30c}}$. 
			\item[(b)] $P(E_2|E_1,F)\geq n^{-2c}$ and $P(W|E_1,E_2,F)\geq n^{-c}$.
			\item[(c)] For every $i\notin S$ and $j\in[k]$, it holds that
			 \[P((X_i,Y_i,Z_i)=(0,0,1)\wedge f_{i,j}(X)=0\wedge g_{i,j}(Y)=0|E_1,E_2,F)\leq n^{-7c}.\]
		\end{itemize}
	\end{lemma}
	\begin{proof}
		Initially let $S=\varnothing$ and $E_1=\mc{X}^{\otimes n},E_2=\mc{Y}^{\otimes n}$. We iterate the process described below to update $S,F,E_1$ and $E_2$ until requirement (c) is met. During the process, we examine the potential function
		\begin{align*}
			\Phi(E_1,E_2,F) &=\log P(W|E_1,E_2,F)+\frac{1}{2}\log P(E_2|E_1,F) \\
			&= \log P(W,E_2|E_1,F)-\frac{1}{2}\log P(E_2|E_1,F),
		\end{align*}
		and ensure that the potential function $\Phi(E_1,E_2,F)$ strictly increases for each iteration. Notice that initially we have \[\Phi(E_1,E_2,F)=\log P(W)\geq -c\log n.\]
		And as long as $\Phi(E_1,E_2,F)\geq -c\log n$, requirement (b) is always satisfied.
		\begin{enumerate}
			\item Let $i\notin S,j\in[k]$ be a coordinate such that requirement (c) is violated, that is
			\[P((X_i,Y_i,Z_i)=(0,0,1)\wedge f_{i,j}(X)=0\wedge g_{i,j}(Y)=0|E_1,E_2,F)> n^{-7c},\]
			which, with the help of requirement (b), implies that
			\begin{equation}\label{eq:3}
				P((X_i,Y_i,Z_i)=(0,0,1)|E_1,F) > n^{-9c},
			\end{equation}
			\begin{equation}\label{eq:4}
				P(f_{i,j}(X)=0|E_1,F,(X_i,Y_i,Z_i)=(0,0,1)) > n^{-9c},
			\end{equation}
			\begin{equation}\label{eq:5}
				P(g_{i,j}(Y)=0|E_1,E_2,F,(X_i,Y_i,Z_i)=(0,0,1),f_{i,j}(X)=0) > n^{-7c}.
			\end{equation}
			Add $i$ to the set $S$. The process stops if no such coordinate $(i,j)$ exists.
			
			\item Apply \cref{lemma:potential} on $(X_i,Y_i,Z_i)$ over the distribution $P$ conditioned on $E_1\wedge F$, with $\varepsilon=n^{-18c}$ and $\delta=n^{-9c}$. Since $P((X_i,Y_i,Z_i)=(0,0,1)|E_1,E_2, F) > 0$, by \eqref{eq:3} we have either
			\[\Phi(E_1,E_2,F\wedge (X_i,Y_i,Z_i)=(0,0,1))\geq \Phi(E_1,E_2,F)-n^{-18c},\]
    		in which case we update $F$ to $F\wedge (X_i,Y_i,Z_i)=(0,0,1)$ and proceed to step 3; Or there exists $(x,y,z)\in\{(1,0,0),(0,1,0)\}$ such that $P((X_i, Y_i, Z_i) = (x,y,z) | E_1, E_2, F) > 0$, and
			\[\Phi(E_1,E_2,F\wedge (X_i,Y_i,Z_i)=(x,y,z))\geq \Phi(E_1,E_2,F)+\frac{1}{8}n^{-27c},\]
			\[P((X_i,Y_i,Z_i)=(x,y,z)|E_1,F)\geq 2^{2\Phi(E_1,E_2,F)}\cdot \frac{1}{8}n^{-36c},\]
			in which case we update $F$ to $F\wedge (X_i,Y_i,Z_i)=(x,y,z)$ and iterate back from step 1.
			
			\item Apply \cref{lemma:potential} on $f_{i,j}(X)$ over the distribution $P$ conditioned on $E_1\wedge F$, with $\varepsilon=n^{-8c}$ and $\delta=n^{-9c}$. Since $P(f_{i,j}(X)=0|E_1,E_2,F) > 0$, by \eqref{eq:4} we have either
			\[\Phi(E_1\wedge f_{i,j}(X)=0,E_2,F)\geq \Phi(E_1,E_2,F)-n^{-8c},\]
			in which case we update $E_1$ to $E_1\wedge f_i(X)=0$ and proceed to step 4; Or we have $P(f_{i,j}(X)=1|E_1,E_2, F) > 0$, and 
			\[\Phi(E_1\wedge f_{i,j}(X)=1,E_2,F)\geq \Phi(E_1,E_2,F)+\frac{1}{8}n^{-17c},\]
			\[P(f_{i,j}(X)=1|E_1,F)\geq 2^{2\Phi(E_1,E_2,F)}\cdot \frac{1}{4}n^{-26c},\]
			in which case we update $E_1$ to $E_1\wedge f_i(X)=1$ and iterate back from step 1.
			
			\item Update $E_2$ to $E_2\wedge g_{i,j}(Y)=1$ and iterate back from step 1. Now that $F$ implies $(X_i,Y_i,Z_i)=(0,0,1)$ and $E_1$ implies $f_{i,j}(X)=0$, by the definition of the game (\cref{def:game}) we know that $W$ implies $g_{i,j}(Y)=1$. Therefore, by \eqref{eq:5}, the increment of potential function in this step is
			\begin{align*}
				&\Phi(E_1,E_2\wedge g_{i,j}(Y)=1,F)-\Phi(E_1,E_2,F)\\
				=\ & \frac{1}{2}\log P(E_2|E_1,F)-\frac{1}{2}\log P(E_2\wedge g_{i,j}(Y)=1|E_1,F) \\
				=\ &-\frac{1}{2}\log P(g_{i,j}(Y)=1|E_1,E_2,F) \\
				\geq\ & \frac{1}{2}P(g_{i,j}(Y)=0|E_1,E_2,F)\\
				\geq\ & \frac{1}{2}n^{-7c}.
			\end{align*}
		\end{enumerate}
		
		Depending on the choices, in each iteration the potential function increases by at least either $\frac{1}{8}n^{-27c}$, or $\frac{1}{8}n^{-17c}-n^{-18c}$, or $\frac{1}{2}n^{-7c}-n^{-8c}-n^{-18c}$, which are all lower bounded by $\frac{1}{8}n^{-27c}$. This means that the potential function is indeed strictly increasing in each iteration, and thus requirement (b) is met. Since it always holds $\Phi(E_1,E_2,F)\leq 0$, this also means that the process will eventually stop, and the total number of iterations is at most $8n^{27c}\cdot c\log n$. In other words, $|S|\leq 8n^{27c}\cdot c\log n\leq n^{28c}$.
		
		Finally, in order to bound $P(E_1|F)$, we prove in below that $P(E_1|F)$ gets multiplied by at least a factor of $n^{-70c}$ in each iteration. Since initially $P(E_1|F)=1$, this implies that eventually after at most $n^{28c}$ iterations, we have $P(E_1|F)\geq (n^{-70c})^{n^{28c}}\geq e^{-n^{30c}}$. In each iteration, when $F$ gets updated to $F\wedge (X_i,Y_i,Z_i)=(x,y,z)$ for some $(x,y,z)\in\mc{S}$, $P(E_1|F)$ changes by a factor of 
		\begin{align*}
			\frac{P(E_1|F,(X_i,Y_i,Z_i)=(x,y,z))}{P(E_1|F)}
			& \geq  P((X_i,Y_i,Z_i)=(x,y,z)|E_1,F) \\
			& \geq \min\left\{n^{-9c}, 2^{2\Phi(E_1,E_2,F)}\cdot \frac{1}{8}n^{-36c}\right\} \\
			& \geq \frac{1}{8}n^{-38c}.
		\end{align*}
		The last line is because $\Phi(E_1,E_2,F)\geq-c\log n$. Furthermore, if step 3 is executed and $E_1$ gets updated to $E_1\wedge f_{i,j}(X)=b$ for some $b\in\{0,1\}$, $P(E_1|F)$ further changes by a factor of 
		\begin{align*}
			\frac{P(E_1\wedge f_{i,j}(X)=b|F)}{P(E_1|F)}
			& = P(f_{i,j}(X)=b|E_1,F) \\
			& \geq \min\left\{n^{-9c}, 2^{2\Phi(E_1,E_2,F)}\cdot \frac{1}{4}n^{-26c}\right\} \\
			& \geq \frac{1}{8}n^{-30c}.
		\end{align*}
		The last line is because at step 3, $\Phi(E_1,E_2,F)\geq-c\log n-n^{-18c}\geq -2c\log n$. Note that in step 4 only $E_2$ changes and $P(E_1|F)$ does not change. So overall, $P(E_1|F)$ changes by a factor of at least $\frac{1}{8}n^{-38c}\cdot \frac{1}{8}n^{-30c}\geq n^{-70c}$.
	\end{proof}

	Notice that the fixing $F$ is independent of the remaining inputs in $[n]\setminus S$. For the rest of the paper, we change $W$ to the event that $(f,g,h)$ wins the copies of $\mc{G}_k(Q)$ in $[n]\setminus S$, and change $E_1,E_2,f,g,h$ to their relevant restrictions to the copies in $[n]\setminus S$, under the fixing $F$. Since $|S|\leq n^{28c}=o(n)$, by also changing $c$ from $\frac{1}{1000}$ to $\frac{1}{1000}\cdot \frac{\log n}{\log(n-|S|)}<\frac{1}{999}$, we can safely assume that $S=\varnothing$ and remove $F$ from the probability conditions, while the distribution $Q$ remains the same and \cref{lemma:event} still holds. This significantly simplifies the discussions later on.

	\section{Almost Independence of Answers in each Coordinate}\label{section:step2}
	
	Let $E_1,E_2$ be specified as in the previous section. In this section, we prove the following lemma:
	\begin{lemma}\label{lemma:ind}
		For every $i\in [n]$ and $j\in[k]$, at least one of the following holds:
		\[P(X_i=0\wedge f_{i,j}(X)=0|E_1,E_2)\leq n^{-3c},\]
		or
		\[P(Y_i=0\wedge g_{i,j}(Y)=0|E_1,E_2)\leq n^{-3c}.\]
	\end{lemma}

	We prove the above lemma using Fourier analysis. Fix some $i\in[n]$ and $j\in[k]$. Define $a:\{0,1\}^n\rightarrow\{-1,0,1\}$ over the inputs of Alice as follows: For every $x\in\{0,1\}^n$,
		\[a(x)=\left\{\begin{array}{ll}
			0 & \textrm{ if }x \notin E_1, \\
			-1 & \textrm{ if }x \in E_1\textrm{ and }x_i=0\textrm{ and }f_{i,j}(x)=0, \\
			1 & \textrm{ otherwise,}
		\end{array}\right.\]
	and let $b(x)=|a(x)|$. Let $\alpha=\frac{1}{2^n}\sum_x b(x)=\widehat{b}(0^n)$.
	
	\begin{proposition}
		$\alpha\geq e^{-n^{130c}}$.
	\end{proposition}
	\begin{proof}
		Recalling the distribution $Q$ in \cref{def:distr}, we have
		\begin{align*}
			P(E_1) & = \sum_{x\in\{0,1\}^n}\left(\frac{1}{2}-\frac{1}{2}n^{-1+100c}\right)^{|x|}
				\left(\frac{1}{2}+\frac{1}{2}n^{-1+100c}\right)^{n-|x|}b(x)\\
				& \leq (1+n^{-1+100c})^n\cdot\frac{1}{2^n}\sum_{x\in\{0,1\}^n}b(x) \\
				& \leq e^{n^{100c}}\alpha.
		\end{align*}
		Since $P(E_1)\geq e^{-n^{30c}}$, we get $\alpha\geq e^{-n^{130c}}$.
	\end{proof}
	
	\begin{lemma}\label{lemma:a}
		For every event $E\subseteq\mc{Y}^{\otimes n}$ on $Y$ with $P(E)>0$, we have
		\[\left|\E[a(X)|E]-\widehat{a}(0^n)\right|\leq \frac{1}{P(E)}\cdot n^{-1/3}\alpha.\]
	\end{lemma}
	\begin{proof}
		Since $P(X_i=1|Y_i=0)=1/2$, we have
		\begin{align*}
			\E[a(X)|E] &=\sum_{y\in E}\E[a(X)|Y=y]\cdot P(Y=y|E) \\
			&=\sum_{y\in E}\frac{1}{2^{n-|y|}}\sum_{x\cdot y=0}a(x)\cdot P(Y=y|E) \\
			&=\sum_{y\in E}\sum_{u\leq y}\widehat{a}(u)\cdot P(Y=y|E) \\
			&=\sum_{u\in\{0,1\}^n}\widehat{a}(u)\cdot P(Y\geq u|E).
		\end{align*}
		Using \cref{lemma:level} on $a$, with the fact that $\ln(1/\alpha)\leq n^{130c}$, we get
		\begin{align*}
			\left|\E[a(X)|E]-\widehat{a}(0^n)\right| 
			&\leq \sum_{u\neq 0^n} |\widehat{a}(u)|\cdot P(Y\geq u|E) \\
			&\leq \frac{1}{P(E)}\sum_{u\neq 0^n} |\widehat{a}(u)|\cdot P(Y\geq u) \\
			&\leq \frac{1}{P(E)}\sum_{\ell=1}^n (2en\cdot\max\{1,\ln(1/\alpha)\})^{\ell/2}\cdot\alpha\cdot (n^{-1+100c})^\ell \\
			&\leq \frac{1}{P(E)}\cdot n^{-1/3}\alpha. \qedhere
		\end{align*}
	\end{proof}
	With the exact same proof on $b$, we can also get
	\begin{lemma}\label{lemma:b}
		For every event $E\subseteq\mc{Y}^{\otimes n}$ on $Y$ with $P(E)>0$, we have
		\[|P(E_1|E)-\alpha|=\left|\E[b(X)|E]-\widehat{b}(0^n)\right|\leq \frac{1}{P(E)}\cdot n^{-1/3}\alpha.\]
		In particular, when $E=\mc{Y}^{\otimes n}$ we get $P(E_1)\geq (1-n^{-1/3})\alpha.$
	\end{lemma}
	\begin{corollary}\label{col:fourier}
		For every event $E\subseteq\mc{Y}^{\otimes n}$ on $Y$ with $P(E|E_1)>0$, we have
		\[\left|\E[a(X)|E_1,E]-\frac{\widehat{a}(0^n)}{\alpha}\right|\leq\frac{1}{P(E|E_1)}\cdot n^{-1/4}.\]
	\end{corollary}
	\begin{proof} Since $a(x)\neq 0$ only when $x\in E_1$, we have 
		\begin{align*}
			\left|\E[a(X)|E_1,E]-\frac{\widehat{a}(0^n)}{\alpha}\right| & =
			 \left|\frac{\E[a(X)|E]}{P(E_1|E)}-\frac{\widehat{a}(0^n)}{\alpha}\right| \\
			 &\leq \left|\frac{\E[a(X)|E]}{P(E_1|E)}-\frac{\widehat{a}(0^n)}{P(E_1|E)}\right| +
			 \left|\frac{\widehat{a}(0^n)}{P(E_1|E)}-\frac{\widehat{a}(0^n)}{\alpha}\right| \\
			 &\leq \left|\frac{\E[a(X)|E]}{P(E_1|E)}-\frac{\widehat{a}(0^n)}{P(E_1|E)}\right|
			 +\left|\frac{\alpha}{P(E_1|E)}-1\right| & 
			 (|\widehat{a}(0^n)| \leq \alpha)\\
			 &\leq \frac{2}{P(E_1\wedge E)}\cdot n^{-1/3}\alpha & \hidewidth(\textrm{\cref{lemma:a,lemma:b}}) \\
			 &=\frac{1}{P(E|E_1)}\cdot n^{-1/3}\cdot \frac{2\alpha}{P(E_1)} &\\
			 &\leq \frac{1}{P(E|E_1)}\cdot n^{-1/4}. & 
			 \hidewidth(\textrm{\cref{lemma:b}}) & \qedhere
		\end{align*}
	\end{proof}
	
	\begin{proof}[Proof for \cref{lemma:ind}]
		Suppose that 
		\[P(Y_i=0\wedge g_{i,j}(Y)=0|E_1,E_2)> n^{-3c}.\]
		Let $E$ be the event $E_2\wedge Y_i=0\wedge g_{i,j}(Y)=0$. By argument (c) in \cref{lemma:event}, we have
		\[P(X_i=0\wedge f_{i,j}(X)=0|E_1,E)\leq n^{-4c}.\]
		Therefore $\E[a(X)|E_1,E]\geq 1-2n^{-4c}$. Since $P(E_2|E_1)\geq n^{-2c}$ and $P(E|E_1)= P(E|E_1,E_2)\cdot P(E_2|E_1)\geq n^{-5c}$, by two applications of \cref{col:fourier} (one on the event $E$ and one on the event $E_2$) we have 
		\begin{align*}
			\E[a(X)|E_1,E_2] &\geq \E[a(X)|E_1,E]-\frac{1}{P(E_2|E_1)}\cdot n^{-1/4}-\frac{1}{P(E|E_1)}\cdot n^{-1/4} \\
			&\geq 1-2n^{-4c}-(n^{2c}+n^{5c})\cdot n^{-1/4} \\
			&\geq 1-2n^{-3c}.
		\end{align*}
		This implies that $P(X_i=0\wedge f_{i,j}(X)=0|E_1,E_2)\leq n^{-3c}$.
	\end{proof}
	
	\section{Independence Implies Low Winning Probability}\label{section:step3}
	
	For every $i\in[n]$, let 
		\[G_{1,i}=\Big\{j\in[k]\Mid P(X_i=0\wedge f_{i,j}(X)=0|E_1,E_2)>n^{-3c}\Big\},\]
		\[G_{2,i}=\Big\{j\in[k]\Mid P(Y_i=0\wedge g_{i,j}(Y)=0|E_1,E_2)>n^{-3c}\Big\}.\]
	Then \cref{lemma:ind} implies that $G_{1,i}\cap G_{2,i}=\varnothing$. For each $x,y\in\{0,1\}^n$, let 
	\[B_1(x)=\big\{i\in[n]\mid x_i=0\wedge \exists j\notin G_{1,i}, f_{i,j}(x)=0\big\},\]
	\[B_2(y)=\big\{i\in[n]\mid y_i=0\wedge \exists j\notin G_{2,i}, g_{i,j}(y)=0\big\}.\]
	And for each $z\in\{0,1\}^n$, let
	\[G_1(z)=\{i\in[n]\Mid z_i=0\wedge h_i(z)\in G_{1,i}\},\]
	\[G_2(z)=\{i\in[n]\Mid z_i=0\wedge h_i(z)\in G_{2,i}\}.\]

	\begin{lemma}\label{lemma:fin}
		Suppose $(f,g,h)$ wins on the inputs $(x,y,z)$. Then at least one of the following holds:
		\begin{itemize}
			\item[(a)] $|x|\leq \frac{2}{5}n$ or $|z|\leq \frac{2}{5}n$,
			\item[(b)] $|B_1(x)|\geq n^{1-c}$ or $|B_2(y)|\geq n^{1-c}$,
			\item[(c)] $|B_1(x)|< n^{1-c}$ and $|B_1(x)\cap \one(y)|\geq 4n^{-c}\cdot |y|$,
			\item[(d)] $|G_1(z)|< \frac{1}{4}n-n^{1-c}$ and $|G_1(z)\cap \one(y)|\geq (1-4n^{-c})\cdot |y|$.
		\end{itemize}
	\end{lemma}

	\begin{proof}
		Since $G_{1,i}\cap G_{2,i}=\varnothing$ for every $i$, we know that $G_1(z)\cap G_2(z)=\varnothing$ for every $z$. On the other hand, by the definition of the game (\cref{def:game}), in order to win it must hold
		\begin{align*}
			\one(y)\subseteq G_1(z)\cup B_1(x) & & \textrm{(since $x_{i,h_i(z)}=0$ when $x_i=z_i=0$)} \\
			\one(x)\subseteq G_2(z)\cup B_2(y) & & \textrm{(since $y_{i,h_i(z)}=0$ when $y_i=z_i=0$)}
		\end{align*}
		Now suppose none of the items (a) to (d) holds. Since
		\[|y|\leq |G_1(z)\cap\one(y)|+|B_1(x)\cap\one(y)|,\]
		it implies that $|G_1(z)|\geq \frac{1}{4}n-n^{1-c}$. Therefore we have
		\[|x|\leq |G_2(z)|+|B_2(y)|\leq n-|z|-|G_1(z)|+|B_2(y)|\leq \frac{7}{20}n+2n^{1-c},\]
		which contradicts the fact that $|x|\geq \frac{2}{5}n$.
	\end{proof}
	
	\begin{proposition}\label{prop:chernoff}
		$P\left(\big||X|-n/2\big|\geq n/10\right)\leq e^{-n/200}$. The same holds when replacing $X$ with $Z$.
	\end{proposition}
	\begin{proof}
		This is a direct application of the Chernoff Bound (\cref{lemma:chernoff}).
	\end{proof}
	
	\begin{lemma}\label{lemma:intx}
		Let $m\geq n^{1-c}$, and $M:\{0,1\}^n\rightarrow 2^{[n]}$ satisfies $M(x)\cap\one(x)=\varnothing$ for all $x\in\{0,1\}^n$. Then we have
		\[P\left(|M(X)|< m\wedge |M(X)\cap \one(Y)|\geq \frac{4m}{n}\cdot |Y|\right)\leq e^{-n^{90c}}.\]
		And the same holds when replacing $X$ with $Z$.
	\end{lemma}
	\begin{proof}
		 Fix an $x\in\{0,1\}^n$ with $|x|\leq \frac{3}{5}n$ and $|M(x)|\leq m$. By \cref{prop:chernoff}, this makes for a probability of $P\big(|X|>\frac{3}{5}n\big)\leq e^{-n/200}$.
		 
		 Since $p=P(Y_i=1|X_i=0)>n^{-1+100c}$, by applying Chernoff Bound on the sets $[n]\setminus\one(x)$ and $M(x)$ respectively, we have
		\[P\left(|Y|\leq \frac{np}{3}\Mid X=x\right)\leq e^{-np/180}<e^{-n^{100c}/180},\]
		\[P\left(|M(x)\cap\one(Y)|\geq \frac{4mp}{3}\Mid X=x\right)\leq e^{-mp/27}<e^{-n^{99c}/27}.\]
		Therefore by union bound,
		\begin{align*}
			P\left(|M(X)|< m\wedge |M(X)\cap \one(Y)|\geq \frac{4m}{n}\cdot |Y|\right)
			& \leq e^{-n/200}+e^{-n^{100c}/180}+e^{-n^{99c}/27} \\
			&\leq e^{-n^{90c}}. \qedhere
		\end{align*}
	\end{proof}
	
	Now we can bound the probability for each item in \cref{lemma:fin}, conditioned on $E_1\wedge E_2$. Recall that $P(E_1\wedge E_2)\geq e^{-n^{30c}}n^{-2c}$ by \cref{lemma:event}.
	\begin{itemize}
		\item[(a)] By \cref{prop:chernoff} we have
		\[P\big(|X|\leq 2n/5| E_1,E_2\big)\leq \frac{1}{P(E_1\wedge E_2)}\cdot e^{-n/200}\leq e^{-n/300}.\]
		Similarly we have $P\big(|Z|\leq 2n/5| E_1,E_2\big)\leq e^{-n/300}$.
		\item[(b)] For each $i\in[n]$, by the definitions of $G_{1,i},G_{2,i}$ and $B_1(x),B_2(x)$, using the union bound over $j\in[k]$ we get
		\[P(i\in B_1(X)|E_1,E_2)\leq kn^{-3c},\quad 
		P(i\in B_2(Y)|E_1,E_2)\leq kn^{-3c}.\]
		Therefore we can bound the expectations of $|B_1(X)|$ and $|B_2(Y)|$:
		\[\E\big[|B_1(X)|\big| E_1,E_2\big]\leq kn^{1-3c},\quad
		\E\big[|B_2(Y)|\big| E_1,E_2\big]\leq kn^{1-3c}.\]
		Thus by Markov's inequality we have
		\[P\big(|B_1(X)|\geq n^{1-c}\big|E_1,E_2\big)\leq kn^{-2c},\quad
		P\big(|B_2(Y)|\geq n^{1-c}\big|E_1,E_2\big)\leq kn^{-2c}.\]
		\item[(c)] Applying \cref{lemma:intx} on $B_1(X)$ with $m=n^{1-c}$, we have
		\begin{align*}
			&P\Big(|B_1(X)|<n^{1-c}\wedge |B_1(X)\cap \one(Y)|\geq 4n^{-c}\cdot |Y|\Mid E_1,E_2\Big) \\
			\leq\ & \frac{1}{P(E_1\wedge E_2)}\cdot e^{-n^{90c}}\leq e^{-n^{80c}}.
		\end{align*}
		\item[(d)] Same as (c), but applying \cref{lemma:intx} on $G_1(Z)$ with $m=\frac{1}{4}n-n^{1-c}\geq n^{1-c}$, we get
		\[P\Big(|G_1(Z)|<\frac{1}{4}n-n^{1-c} \wedge |G_1(Z)\cap \one(Y)|\geq (1-4n^{-c})\cdot |Y|\Mid E_1,E_2\Big)\leq e^{-n^{80c}}.\]
	\end{itemize}
	Putting everything together by a union bound, we get
	\[P(W|E_1,E_2)\leq 2e^{-n/300}+2kn^{-2c}+2e^{-n^{80c}}<n^{-c},\]
	as $k$ is a constant and $n$ is sufficiently large. This leads to a contradiction to the result (b) in \cref{lemma:event}, which refutes the assumption in \cref{claim:main}, and thus proves \cref{thm:main}.
	\qedhere
	
	\paragraph{Acknowledgements.} We thank Justin Holmgren for important conversations and collaboration in early stages of this work.
	
	\bibliography{main}
	\bibliographystyle{alpha}

\end{document}